\documentclass[twocolumn,showpacs,notitlepage,nofootinbib,floatfix]{revtex4-1} 
\usepackage{graphicx}
\usepackage{amssymb}
\usepackage{amsmath}
\usepackage{amsthm}
\usepackage{amsfonts}
\usepackage{mathrsfs}
\usepackage{color}
\usepackage{hyperref}
\usepackage{mathtools}

\newcommand{\col}[1]{{\mathrm{color}(#1)}}

\newcommand{\ket}[1]{|#1\rangle}
\newcommand{\cnot}{\text{CNOT}}

\theoremstyle{lemma}
\newtheorem{lemma}{Lemma}

\theoremstyle{cond}
\newtheorem{cond}{Condition}

\begin{document}

\title{Universal transversal gates with color codes --- a simplified approach}

\author{Aleksander Kubica} \affiliation{Institute for Quantum Information \& Matter, California Institute of Technology,  Pasadena CA 91125, USA}
\author{Michael E. Beverland} \affiliation{Institute for Quantum Information \& Matter, California Institute of Technology,  Pasadena CA 91125, USA}

\date{\today}

\begin{abstract}
We provide a simplified, yet rigorous presentation of the ideas from Bomb\'{i}n's paper \emph{Gauge Color Codes} \cite{Bombin2013}. Our presentation is self-contained, and assumes only basic concepts from quantum error correction. We provide an explicit construction of a family of color codes in arbitrary dimensions and describe some of their crucial properties. Within this framework, we explicitly show how to transversally implement the generalized phase gate $R_n = \text{diag}(1, e^{2\pi i/2^n})$, which deviates from the method in Ref.~\cite{Bombin2013}, allowing an arguably simpler proof. We describe how to implement the Hadamard gate $H$ fault-tolerantly using code switching. In three dimensions, this yields, together with the transversal $\cnot$, a fault-tolerant universal gate set $\{H,\cnot,R_3 \}$ without state-distillation.
\end{abstract}
\maketitle

\section{Introduction}

To build a fully functioning quantum computer, it is necessary to encode quantum information to protect it from noise. In physical systems, one expects noise to act locally. Therefore, \emph{topological codes} \cite{Kitaev03,LevinWen,Bombin2006,BDFN10}, which naturally protect against local errors, represent our best hope for storing quantum information. However, a quantum computer must also be capable of processing this information. This motivates the search for topological codes allowing the implementation of a set of gates which (i) can operate in the presence of typical noise without corrupting the stored information, and (ii) can perform any computation on the encoded information. A theoretical framework has been developed around these ideas --- a gate which is \emph{fault-tolerant} does not propagate typical errors into uncorrectable errors~\cite{Shor1996,Preskill1998}, and therefore satisfies (i). A set of gates which is \emph{universal} can generate any unitary on the code space with arbitrary precision \cite{Kitaev1997, Nielsen2010}, and therefore satisfies (ii).

The known methods of implementing a universal, fault-tolerant gate set in topological codes typically require an enormous amount of overhead. For instance, magic state distillation \cite{Bravyi2005} with the two-dimensional toric code requires many additional ancilla qubits \cite{Fowler2012}, whereas computing by braiding non-abelian anyons \cite{Kitaev03,Nayak2008} requires additional time to move anyons around macroscopic loops \cite{Beckman2001}. These forms of overhead can make quantum processing orders of magnitude less efficient than storage alone in topological codes. This may render such approaches impractical given the experimental difficulty of scaling up quantum hardware~\cite{Fowler2012,Devoret2013,Wecker2014}. In this paper we focus on a new construction by Bomb\'in~\cite{Bombin2013}, for a universal fault-tolerant gate set with topological color codes. This seems not to involve significant additional overhead, however a lattice of at least three dimensions is required, limiting the construction's practicality for reasons of architecture. 

Following Bomb\'in's construction, we use the simplest form of fault-tolerant gate --- the \emph{transversal gate}, which is a code-space preserving unitary composed of separate unitaries applied to each physical qubit. However, according to a no-go theorem by Eastin and Knill \cite{Eastin2009}, for any code which protects against arbitrary single-qubit errors, the set of transversal gates forms a finite group and therefore cannot be universal. Some recent approaches to circumvent this no-go theorem in order to implement a universal gate set with transversal gates have been put forward~\cite{Jochym2014,Paetznick2013,Anderson2014}. 

In Ref.~\cite{Bombin2013}, Bomb\'in applies the approach of \emph{gauge fixing}~\cite{Paetznick2013,Anderson2014} to color codes in a $d$-dimensional lattice. Color codes were first introduced in two dimensions by Bomb\'{i}n and Martin-Delgado in Ref.~\cite{Bombin2006}. They are \emph{topological stabilizer codes} \cite{Gottesman1996,Calderbank1997,Kitaev03,Bravyi2013}, meaning they are defined on a lattice and have macroscopic distance together with geometrically local stabilizer generators. The main new conceptual contribution in Ref.~\cite{Bombin2013} is that gauge fixing allows one to fault-tolerantly switch between a (stabilizer) color code on a $d$-dimensional lattice, in which $\cnot$ and $R_d=\textrm{diag}\left(1,\exp(\frac{2\pi i}{2^d})\right) $ are transversal, and a different (subsystem) color code on the same lattice, in which $H$ is transversal. Critically, for $d \geq 3$, $\{ H,\cnot, R_d \}$ forms a universal gate set. To the authors' knowledge, this represents the first construction using gauge fixing to achieve a universal gate set in a topological code. 

In Ref.~\cite{Bombin2013}, Bomb\'in argues that for every $d \geq 2$, there exists a  $d$-dimensional color code with a transversal implementation of $R_d \in\mathcal{P}_d\setminus \mathcal{P}_{d-1}$, which is the main technical contribution therein. At the same time, for any topological stabilizer code, Bravyi and K\"{o}nig \cite{Bravyi2013} showed that the group of logical gates implemented transversally must be contained in $\mathcal{P}_d$, the $d^{\text{\,th}}$ level of the Clifford hierarchy\footnote{The Clifford hierarchy is defined sequentially for $j>1$ according to $\mathcal{P}_j = \{ \text{unitary }U | U P U^\dagger \in \mathcal{P}_{j-1} ~\forall P \in \mathcal{P}_{1} \}$, with $\mathcal{P}_1$ representing the Pauli group. Note that $\mathcal{P}_2$ is the well-known Clifford group.}~\cite{Gottesman1999}. These results have been extended beyond the stabilizer code setting \cite{Pastawski2014,Beverland2014}. Color codes are the only family of topological stabilizer codes currently known to saturate the Bravyi-K\"{o}nig classification in every dimension $d\geq 2$. 

In this paper, we provide a simplified yet rigorous presentation of the ideas in Ref.~\cite{Bombin2013}. The organization is as follows. First, to build some intuition, we introduce color codes in two dimensions in Section~\ref{sec:2dim}. We explain how to transversally implement the gate set $\{H,\cnot,R_2 \}$, which generates the Clifford group. Then, we describe the generalization of color codes to  $d$ dimensions in Section~\ref{sec:Ddim}. Next, in Section~\ref{sec:gates} we discuss transversal gates in those codes with an emphasis on the phase gate $R_n$, and show that in certain  $d$-dimensional color codes $R_d$ is transversal. Our construction utilizes the bipartite property of the lattice allowing for a simpler verification than in Ref.~\cite{Bombin2013}. Finally, in Section~\ref{sec:universal} we explain how to switch between color codes fault-tolerantly using the technique of gauge fixing. In particular, this allows one to implement a fault-tolerant universal gate set $\{H,\cnot,R_3 \}$ in a color code in three dimensions.

\section{Color code in two dimensions}
\label{sec:2dim}

In this section, we give an explicit construction of a stabilizer color code in two dimensions \cite{Bombin2006,Bombinbook}. We consider a $3$-valent lattice formed as a tiling of a sphere, such that faces of the lattice are colored with three colors, where neighboring faces have distinct colors. Qubits are placed at the vertices of this lattice. To define a color code on this lattice, we associate an $X$- and a $Z$-type stabilizer generator with every face. This code encodes no logical qubits. A new code, which encodes a single logical qubit, can be formed through the removal of a single physical qubit. We describe the transversal implementation of the logical gates $\overline{\cnot}$, $\overline{H}$ and $\overline{R}_2$ in the new code\footnote{We use a bar to indicate action on logical code space. The absence of a bar indicates action on physical qubits.}.

\subsection{Color code with no encoded qubits}
\label{ColorCode2dNoQubits}

Color codes in two dimensions are CSS stabilizer codes \cite{Gottesman1996,Calderbank1997}, and are therefore specified by their stabilizer group $\mathcal{S}$ generated by $X$- and $Z$-type stabilizer generators. The code space is the simultaneous $+1$ eigenspace of every stabilizer generator. In the construction, we use a two-dimensional lattice $\mathcal{L}_0^*$, obtained from a tiling of the 2-sphere, and satisfying the following requirements
\begin{itemize}
\item valence --- every vertex is 3-valent, meaning it belongs to exactly 3 edges,
\item colorability --- faces can be colored with 3 colors: red, green and blue, such that every two faces sharing an edge have different colors.
\end{itemize}
An example of such a tiling of the 2-sphere is presented in Fig.~\ref{fig:2DLatticeConstruction}(a). From these properties alone, one can show that the total number of vertices in $\mathcal{L}_0^*$ is even. To see this, note that the Euler characteristic gives $V-E+F=2$, where $V$, $E$ and $F$ denote the number of vertices, edges and faces in $\mathcal{L}^*_0$, respectively. Since every vertex is $3$-valent, we obtain $E=\frac{3}{2} V$, and then $V=2(F-2)$, which is even.

\begin{figure}[h!]
\includegraphics[width=0.45\textwidth]{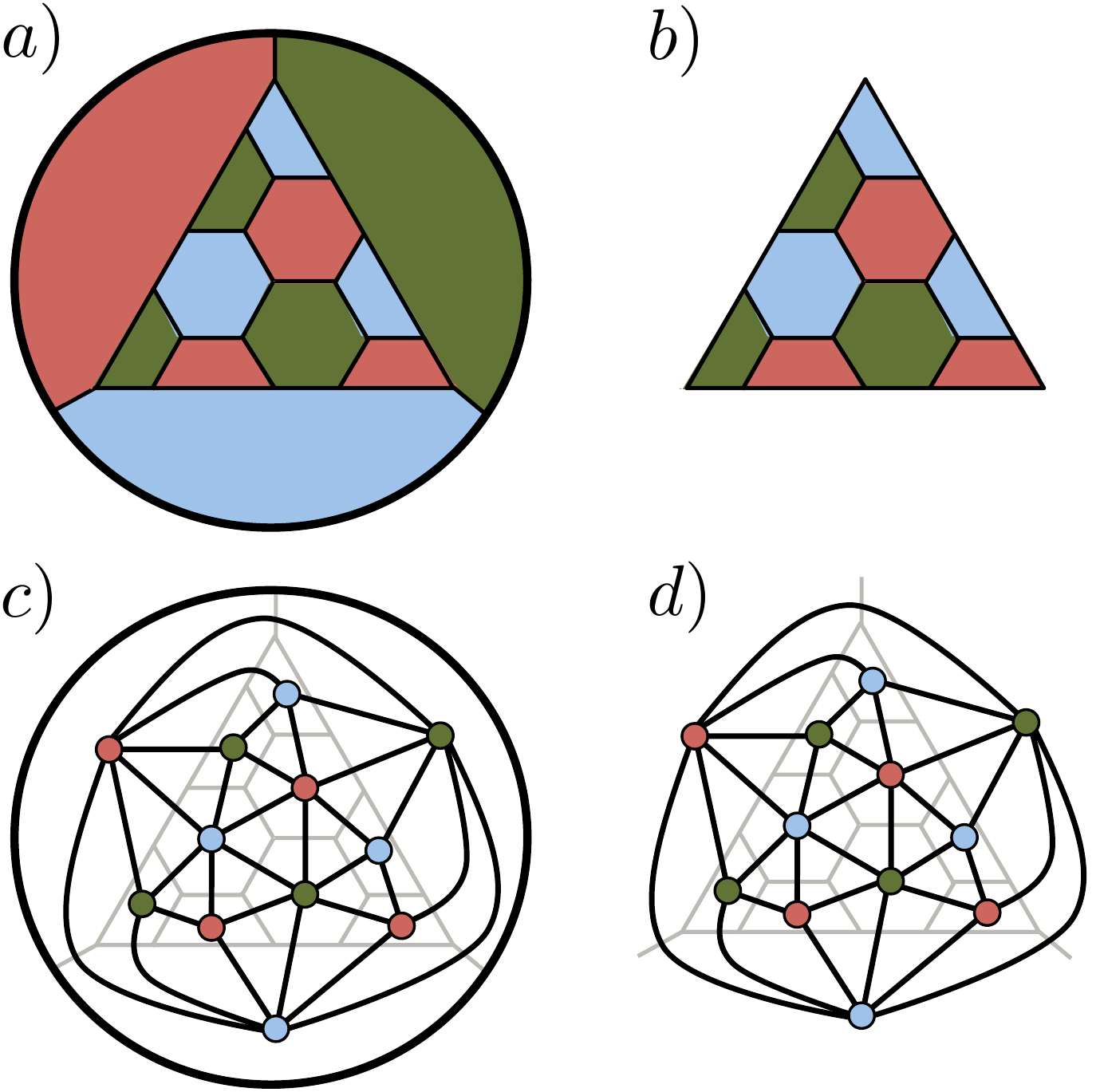}
\caption{Construction of color codes in two dimensions. In (a) and (b), qubits are placed at vertices, and $X$- and $Z$-type stabilizer generators are associated with faces. In (c) and (d) (the dual picture), qubits are placed on faces, and $X$- and $Z$-type stabilizer generators are associated with vertices. (a) Take a lattice $\mathcal{L}_0^*$, which is a tilling of the 2-sphere with 3-colorable faces and 3-valent vertices. The surrounding circle is identified with a vertex $v$. The color code on $\mathcal{L}_0^*$ encodes no logical qubits. (b) To obtain $\mathcal{L}^*$, remove from $\mathcal{L}_0^*$ the vertex $v$, together with the three edges and three faces containing it. The color code on $\mathcal{L}^*$ encodes one logical qubit. (c) (Dual) lattice $\mathcal{L}_0$ is obtained from $\mathcal{L}_0^*$ by replacing faces, edges and vertices by vertices, edges and faces, respectively. All faces are triangles, and the vertices are 3-colorable. The color code on $\mathcal{L}_0$ encodes no logical qubits. (d) Lattice $\mathcal{L}$ formed from $\mathcal{L}_0$ by removing a single face. No stabilizer generators are associated with those vertices belonging to the boundary of the removed face. The color code on $\mathcal{L}$ encodes one logical qubit.}
\label{fig:2DLatticeConstruction} 
\end{figure} 

At every vertex in $\mathcal{L}^*_0$ we place a qubit. We refer to the set of all qubits by $Q$, whereas by $\mathcal{Q}(\Pi)\subset Q$ we denote the set of vertices of a face $\Pi$. Alternatively, we can think of $\mathcal{Q}(\Pi)$ as the set of qubits belonging to $\Pi$. To define the color code, it is sufficient to specify $X$- and $Z$-type stabilizer generators. For every face $\Pi$, we define an $X$-type stabilizer generator $X(\Pi)$ to be a tensor product of Pauli $X$ operators supported on qubits $\mathcal{Q}(\Pi)$, similarly for $Z$-type generators. Then, the stabilizer group $\mathcal{S}$ is generated by 
\begin{equation}
\mathcal{S}=\langle X(\Pi),Z(\Pi)\text{, for every face $\Pi$ in $\mathcal{L}_0^*$} \rangle .
\end{equation}

To prove that this specifies a well-defined stabilizer code, we must verify that all the generators of $\mathcal{S}$ commute. It is sufficient to check that for any two faces $\Pi_1$ and $\Pi_2$ in $\mathcal{L}_0^*$, $X(\Pi_1)$ and $Z(\Pi_2)$ commute. First take the case $\Pi_1\neq\Pi_2$. If $\Pi_1$ and $\Pi_2$ share no vertices, then $X(\Pi_1)$ and $Z(\Pi_2)$ trivially commute. If they share a vertex, then by $3$-valence, they also share an edge. Moreover, due to $3$-colorability, $\Pi_1$ and $\Pi_2$ cannot share two consecutive edges, and thus their intersection has to contain an even number of vertices,
\begin{equation}
|\mathcal{Q}(\Pi_1)\cap \mathcal{Q}(\Pi_2)|\equiv 0 \mod 2.
\end{equation}
For the case $\Pi_1=\Pi_2 = \Pi$, due to $3$-colorability and $3$-valence, the number of vertices belonging to a face $\Pi$ is even,
\begin{equation}
\label{EvenVerticesOfFaces}
|\mathcal{Q}(\Pi)|\equiv 0 \mod 2.
\end{equation}
Therefore, we obtain commutation of $X(\Pi_1)$ and $Z(\Pi_2)$ for arbitrary $\Pi_1$ and $\Pi_2$.

From the construction of the lattice, one obtains that each vertex belongs to exactly three faces, colored with three different colors. Thus, one can express the set of vertices in $\mathcal{L}_0^*$ as the disjoint union\footnote{We use the \emph{disjoint union} $A \sqcup B$ in place of the union $A \cup B$ of two sets $A$ and $B$ when their instrsection is empty, $A\cap B=\emptyset$.} of vertices belonging to red faces, and similarly for green and blue~\cite{Bombin2006,Bombinbook}, namely
\begin{equation}
Q = \bigsqcup_{\Pi_R} \mathcal{Q}(\Pi_R) = \bigsqcup_{\Pi_G} \mathcal{Q}(\Pi_G) = \bigsqcup_{\Pi_B} \mathcal{Q}(\Pi_B),
\end{equation}
where $\{ \Pi_R \}$, $\{ \Pi_G \}$ and $\{ \Pi_B \}$ are the sets of all red, green and blue faces, respectively. This implies that not all the stabilizer generators we have defined are independent
\begin{eqnarray}
\prod_{\Pi_R} X(\Pi_R) &=& \prod_{\Pi_G} X(\Pi_G) = \prod_{\Pi_B} X(\Pi_B),
\label{eq:conditionsontiling1}\\
\prod_{\Pi_R} Z(\Pi_R) &=&\prod_{\Pi_G} Z(\Pi_G) = \prod_{\Pi_B} Z(\Pi_B).
\label{eq:conditionsontiling2}
\end{eqnarray}
In fact, these are the only conditions~\cite{Bombin2007,Bombinbook} which relate the stabilizer generators to one another. 

We can now verify that the color code which we have defined on the lattice $\mathcal{L}_0^*$ encodes no logical qubits. As before, using the Euler characteristic we obtain $F-2=E-V$, and from $3$-valence of vertices --- $E=\frac{3}{2} V$. We have placed physical qubits at vertices, thus $|Q|=V$. There are $2F-4$ independent stabilizer generators, since there are two stabilizer generators for every face and four conditions~(\ref{eq:conditionsontiling1})~and~(\ref{eq:conditionsontiling2}). The number of logical qubits is equal to the number of physical qubits minus the number of independent stabilizer generators, and we obtain
\begin{equation}
|Q|-(2F-4)=V-2(E-V)=0.
\end{equation}

\subsection{Color code with one logical qubit}
\label{2dOneQubit}

To obtain a color code with one encoded logical qubit, we can remove one vertex from the lattice $\mathcal{L}_0^*$, together with three edges and three faces it belongs to, obtaining a new lattice $\mathcal{L}^*$ (see Fig.~\ref{fig:2DLatticeConstruction}b). By removing one vertex, we also discard six stabilizer generators associated with the removed faces, and thus the stabilizer generators no longer have to satisfy (\ref{eq:conditionsontiling1}) and (\ref{eq:conditionsontiling2}). One can check that this new code encodes one logical qubit, since there is one qubit more than independent stabilizer generators. By removing more vertices, one could encode more logical qubits, but we will not analyze that case. Note that the total number of qubits in $\mathcal{L}^*$ is odd, $|Q|\equiv 1 \mod 2$,  which plays an important role in our considerations.

On physical grounds, it is of interest to consider stabilizer codes with stabilizer generators which are low-weight and geometrically local. In the construction we have presented, this can be achieved if each face in the lattice $\mathcal{L}^*$ is geometrically local and contains a small number of vertices, as in Fig.~\ref{fig:2DLatticeConstruction}b. It can be shown that following this construction, the resulting color code has macroscopic distance \cite{Bombin2006}, and therefore is a topological stabilizer code.

Later, when we discuss color codes in $d$ dimensions, we follow a similar construction. We briefly outline the procedure here, deferring detailed discussion to Section~\ref{sec:Ddim}. We start with a tiling of a  $d$-sphere, place qubits at vertices and define (gauge group) generators to be supported on suitable cells. Then, we remove one vertex and all the cells containing it. In particular, we discard generators supported on the removed cells. Such a code encodes only one logical qubit \cite{Bombin2007}.

\subsection{Transversal gates}

In this paper we consider stabilizer codes encoding only one logical qubit, with the stabilizer group $\mathcal{S}$. In this setting, a \emph{transversal gate} $\overline{U}$ on a single logical qubit is implemented as a tensor product of single physical qubit unitaries $U_1\otimes \ldots \otimes U_{|Q|}$, which preserves the code space. On the other hand, a logical gate on two logical qubits requires two copies of the code, in which case we say that the \emph{overall} code space is the $+1$ eigenspace of the elements in $\mathcal{S}\otimes \mathcal{S}$. A transversal gate on two logical qubits is implemented as a tensor product of two qubit gates on pairs of corresponding qubits in both copies of the code, which preserves the overall code space. Observe that transversal gates are fault-tolerant since they do not spread errors within each copy of the code.

We now show that in the two-dimensional color code described in the previous subsection, one can transversally implement the gate set $\{ \overline{H},\overline{\cnot}, \overline{R}_2 \}$, which generates the (non-universal) Clifford group. The Clifford group, combined with computational basis state preparation and measurement, can be simulated efficiently on a classical computer~\cite{Gottesman1998, Aaronson2004}. For each gate, $\overline{H}$, $\overline{\cnot}$ and $\overline{R}_2$, we verify that a particular transversal unitary implements the logical gate by showing that it has the correct action under conjugation on generators of the logical Pauli group, and that the stabilizer group is preserved\footnote{Preservation of the stabilizer group is a sufficient (but not necessary) condition that implies preservation of the code.}. 

The two-dimensional color code is a CSS stabilizer code encoding a single logical qubit with logical Pauli operators $\overline{X} = X(Q)$ and $\overline{Z} = Z(Q)$. In addition it is a \emph{self-dual CSS stabilizer code} --- a code with the same support for $X$- and $Z$-type stabilizer group elements (for each face, there is an $X$- and a $Z$-type generator). This implies that the logical Hadamard gate can be implemented transversally, as under conjugation by $H(Q)$, $\overline{X} \mapsto H(Q) X(Q)H(Q)^\dagger = \overline{Z}$ and similarly $\overline{Z} \mapsto \overline{X}$. Moreover, $X(\Pi)\mapsto Z(\Pi)$, $Z(\Pi)\mapsto X(\Pi)$, and thus $\mathcal{S}$ is preserved.

The logical gate $\overline{\cnot}$ can be implemented transversally between two identical copies of this color code by applying a physical gate $\cnot$ to every pair of corresponding qubits in the first and the second copy. This can be verified by checking that under conjugation by $\overline{\cnot}$, $\overline{X}\,\overline{I} \mapsto \overline{X}\,\overline{X}$, $\overline{I}\,\overline{X} \mapsto \overline{I}\,\overline{X}$, $\overline{Z}\,\overline{I} \mapsto \overline{Z}\,\overline{I}$, $\overline{I}\,\overline{Z} \mapsto \overline{Z}\,\overline{Z}$ and $\mathcal{S}\otimes\mathcal{S}$ is preserved\footnote{Notice that generators of $\mathcal{S}\otimes\mathcal{S}$ are mapped under conjugation to a different generators, namely $X(\Pi)\otimes I(\Pi) \mapsto X(\Pi)\otimes X(\Pi)$, $Z(\Pi)\otimes I(\Pi) \mapsto I(\Pi)\otimes Z(\Pi)$, $I(\Pi)\otimes X(\Pi) \mapsto I(\Pi)\otimes X(\Pi)$ and $I(\Pi)\otimes Z(\Pi) \mapsto Z(\Pi)\otimes Z(\Pi)$.}.

To show that $\overline{R}_2$ can be implemented transversally, we use the fact that the set of vertices in $\mathcal{L}^*$ is bipartite (see Fig.~\ref{fig:BipartiteGraph}(a)). In other words, $Q$ can be split into two subsets, $T$ and $T^c:=Q\setminus T$, such that vertices in $T$ are connected only to vertices in $T^c$ and vice versa. To prove this, first note that every face in $\mathcal{L}_0^*$ has an even number of edges. Moreover, every cycle in $\mathcal{L}_0^*$ (as a tiling of the 2-sphere) is contractible. This implies that every cycle in $\mathcal{L}_0^*$ is a boundary of faces and is therefore even. Using the following lemma
\begin{lemma}[Graph Bipartition]
A graph containing only even cycles is bipartite~\cite{Wilson1996}.
\label{lemma:Bipartition}
\end{lemma}
\noindent we see that $\mathcal{L}_0^*$ must be bipartite, and so is the lattice $\mathcal{L}^*$ due to its construction from $\mathcal{L}_0^*$.

Now, we can show that $R= R^k_2 (T) R^{-k}_2 (T^{c})$ implements $\overline{R}_2$, for some choice of integer $k$. We use the relations $R_2 X R_2^{\dagger} = i X Z$ and $R_2 Z R_2^{\dagger} = Z$. Since $|Q|\equiv 1 \mod 2$, then $|T|-|T^{c}| = 2|T|-|Q|\equiv \pm 1 \mod 4$, and picking $k= |T|-|T^{c}| \mod 4$ ensures that $k(|T|-|T^{c}|) \equiv 1 \mod 4$. With this choice of $k$, the action by conjugation of $R= R^k_2 (T) R^{-k}_2 (T^{c})$ on the logical $\overline{X}$ and $\overline{Z}$ is 
\begin{eqnarray}
R \overline{X}\, R^\dag &=& i^{k(|T|-|T^{c}|)} \overline{X}\,\overline{Z}=i \overline{X}\,\overline{Z},\\
R \overline{Z}\, R^\dag &=& \overline{Z}.
\end{eqnarray}

Furthermore, as every face $\Pi$ in the lattice $\mathcal{L}^*$ has an equal number of vertices in $T$ and $T^c$, under the action of $R$ the stabilizer generators $X(\Pi)$ and $Z(\Pi)$ become: 
\begin{eqnarray}
R X(\Pi) R^\dagger &=& i^{k(|T\cap \Pi|-|T^c\cap \Pi|)} X(\Pi) Z(\Pi) \\
&=& X(\Pi) Z(\Pi) \in \mathcal{S},\\ 
R Z(\Pi) R^\dagger &=& Z(\Pi),
\end{eqnarray}
implying that the stabilizer group $\mathcal{S}$ is preserved. This completes the verification that $R$ implements $\overline{R}_2$.

\begin{figure}[h!]
\includegraphics[width=0.48\textwidth]{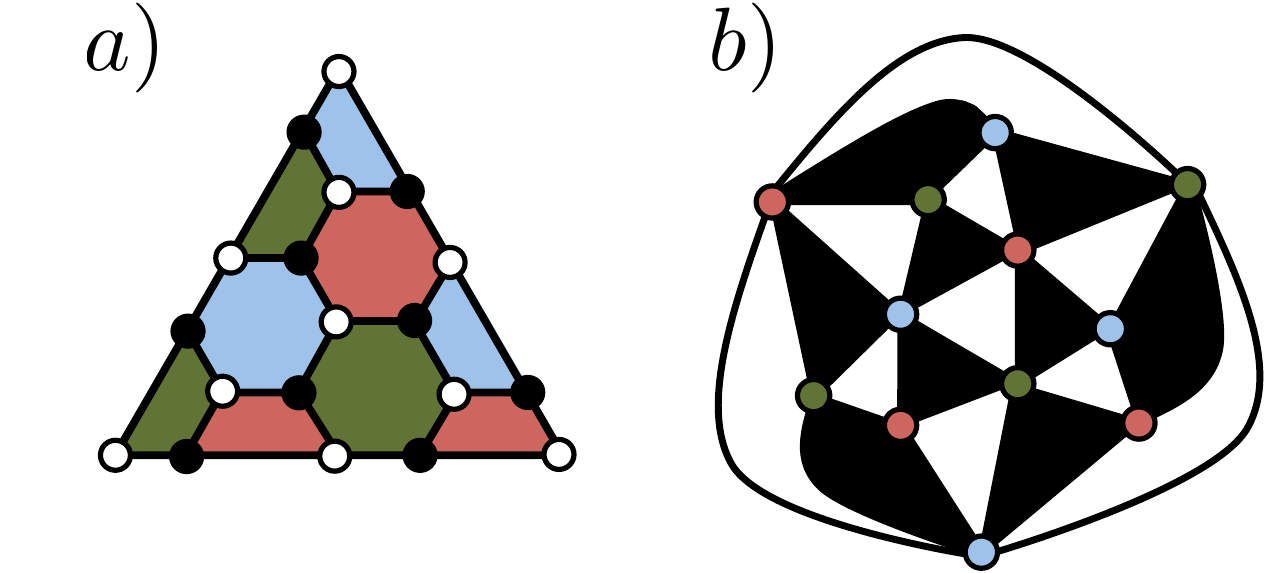}
\caption{(a) The set of vertices of $\mathcal{L}^*$, the lattice used to define the color code, is bipartite --- it can be split into two subsets: $T$ (hollow circles), and its compliment $T^{c}$ (filled circles). Vertices in $T$ are only connected to vertices in $T^c$ and vice versa. The logical gate $\overline{R}_2$ can implemented by applying $R_2^k$ to qubits in $T$, and $R_2^{-k}$ to qubits in $T^{c}$, where $k \equiv  |T|-|T^{c}| \mod 4$. (b) The dual lattice $\mathcal{L}$. Faces are bipartite.}
\label{fig:BipartiteGraph} 
\end{figure} 

\subsection{Dual lattice picture}
\label{DualLattice2d}

We can alternatively express the construction of color codes in the dual lattice picture, which we use extensively in the later discussion for $d>2$ dimensions. We use a two-dimensional (dual) lattice $\mathcal{L}_0$, obtained from a tiling of the 2-sphere, and satisfying the following requirements
\begin{itemize}
\item all faces are triangles,
\item vertices are 3-colorable, meaning two vertices belonging to the same edge are colored with different colors.
\end{itemize}
See Fig.~\ref{fig:2DLatticeConstruction}(c) for a simple example. Note that these conditions are equivalent to the conditions of 3-valence of vertices and 3-colorability of faces required for the tiling $\mathcal{L}_0^*$ of the 2-sphere, where lattices $\mathcal{L}_0^*$ and $\mathcal{L}_0$ are dual to one another.

A qubit is placed on every face of $\mathcal{L}_0$, and an $X$- and a $Z$-type stabilizer generator is associated with every vertex, meaning they are supported on qubits corresponding to faces containing that vertex. The resulting color code is exactly the same as that described in Section~\ref{ColorCode2dNoQubits}, and therefore has zero logical qubits. To encode a single logical qubit, one should remove a face from $\mathcal{L}_0$, together with stabilizer generators associated with the vertices belonging to the removed face, see Fig.~\ref{fig:2DLatticeConstruction}(d).

The bipartition of vertices in $\mathcal{L}^*$ corresponds to a bipartition of faces in $\mathcal{L}$, meaning that faces can be split into two sets, $T$ and its compliment $T^c$, such that faces in $T$ share an edge only with faces in $T^c$ and vice-versa. See Fig.~\ref{fig:BipartiteGraph}(b).

\section{Color code in higher dimensions}
\label{sec:Ddim}
Here we present a construction of color codes on $d$-dimensional lattices. In higher dimensions it is easier to describe the construction in the language of the dual lattice. The majority of this section is devoted to defining dual lattices satisfying certain conditions and analyzing their properties. The discussion is a generalization of that already presented for two dimensions. The basic idea of how to construct the dual lattice $\mathcal{L}$ is to first tile a $d$-sphere with $d$-simplices to form a lattice $\mathcal{L}_0$. We require that every vertex in $\mathcal{L}_0$ can be assigned one of $d+1$ distinct colors and two vertices belonging to the same edge have different colors. The lattice $\mathcal{L}$, used to define the color code, is formed by removing one $d$-simplex from $\mathcal{L}_0$.

\subsection{Simplicial complexes and colorability}

A  $d$-simplex $\delta$ is a  $d$-dimensional polytope which is a convex hull of its $d+1$ affinely independent vertices $v_0,v_1,\ldots, v_d$, namely 
\begin{equation}
\delta=\left\{\sum_{i=0}^d t_i v_i \right|\left. 0\leq t_i  \wedge\sum_{i=0}^d t_i = 1\right\}.
\end{equation}
In particular, $0$-simplices are vertices, $1$-simplices are edges, $2$-simplices are triangles, $3$-simplices are tetrahedra and so on.

A convex hull of a subset of vertices of size $k+1\leq d+1$ is a $k$-simplex $\sigma$, which we call a $k$-face of $\delta$, and $\sigma\subset\delta$. For example, the faces of a $3$-simplex (a tetrahedron) are: four $0$-simpices, six $1$-simplices, four $2$-simplices and a single $3$-simplex. More generally, $\delta$ contains $d+1\choose k+1$ $k$-faces, since every $k$-face is uniquely determined by the choice of $k+1$ vertices spanning it. By $\Delta_k(\delta)$ we call the set of all $k$-faces of $\delta$, namely 
\begin{equation}
\Delta_k(\delta) = \{ \sigma\subset\delta | \sigma\textrm{ is a $k$-simplex}\}.
\end{equation}

Instead of having only one simplex, we can consider a collection of them. Moreover, we can create new objects, called simplicial complexes~\cite{Hatcher2002}, by gluing simplices along their proper faces of matching dimension. We restrict ourselves to simplicial complexes containing finitely many simplices. We will define a $d$-dimensional color code on a lattice $\mathcal{L}$ obtained by gluing together $d$-simplices. The technical name for such a lattice is a homogeneous simplicial $d$-complex.

Although $\mathcal{L}$ is formally a collection of simplices, by the same symbol we also denote the union of these simplices as a topological space. Notice that $\mathcal{L}$ is a manifold with a boundary, which we can think of as being embedded in real space. We denote by $\partial\mathcal{L}$ the set of simplices belonging to the boundary of $\mathcal{L}$, where the boundary of $\mathcal{L}$ is the set of points in the closure of $\mathcal{L}$ not belonging to the interior of $\mathcal{L}$. Moreover, by $\Delta'_k(\mathcal{L})$ we understand a set of all $k$-simplices belonging to $\mathcal{L}\setminus\partial\mathcal{L}$. Note that $\Delta'_d(\mathcal{L})=\Delta_d(\mathcal{L})$.

We say that a simplicial  $d$-complex $\mathcal{L}$ is $(d+1)$-colorable if we can introduce a function 
\begin{equation}
\textrm{color}:\Delta_0 (\mathcal{L})\rightarrow\mathbb{Z}_{d+1},
\end{equation}
where $\mathbb{Z}_{d+1}=\{ 0,1,\ldots, d \}$ is a set of $d+1$ colors, and two vertices belonging to the same edge have different colors. Moreover, by $\col{\delta}$ we understand the set of colors assigned to all the vertices of a simplex $\delta$, namely 
\begin{equation}
\col{\delta}=\bigsqcup_{v\in\Delta_0 (\delta)}\col{v}.
\end{equation}

An example of a $3$-colorable, homogeneous, simplicial $2$-complex is the lattice $\mathcal{L}$ shown in Fig.~\ref{fig:2DLatticeConstruction}(d). Note in particular that it is composed of nineteen $2$-simplices (triangles). The exact shape of objects in $\mathcal{L}$ is not important due to its topological nature --- the lattice is not rigid and can be smoothly deformed. In this example, $\Delta'_0(\mathcal{L})$ consists of the set of $9$ vertices (the three vertices in the boundary are excluded). $\Delta'_1(\mathcal{L})$ is the set of $27$ edges, (the three edges in the boundary are excluded). $\Delta'_2(\mathcal{L})$ is the set of all $19$ triangular faces.

\subsection{Definition of color code}

Here we define color codes on a $d$-dimensional lattice $\mathcal{L}$, which must satisfy the following conditions
\begin{cond}
$\mathcal{L}$ is a homogeneous simplicial  $d$-complex obtained as a triangulation of the interior of a  $d$-simplex.
\label{cond1}
\end{cond}
\begin{cond}
$\mathcal{L}$ is $(d+1)$-colorable.
\label{cond2}
\end{cond}
\noindent One can obtain such a lattice $\mathcal{L}$ from any $(d+1)$-colorable tiling of the $d$-sphere with $d$-simplices, followed by the removal of one $d$-simplex. In $d=2$ dimensions, this is precisely the procedure described in Section~\ref{DualLattice2d}. An explicit construction of a family of lattices satisfying these conditions is outlined in Appendix~\ref{LatticeConstruction}.

Qubits are placed on each and every  $d$-simplex of $\mathcal{L}$, and thus the set of all qubits $Q$ is equal to $\Delta_d(\mathcal{L})$. This motivates the next definition, namely for a simplex $\delta\subset\mathcal{L}\setminus\partial\mathcal{L}$ we define 
\begin{equation}
\mathcal{Q}(\delta)=\{ \sigma\in\Delta_d(\mathcal{L}) | \sigma\supset\delta\}.
\end{equation}
In other words, $\mathcal{Q}(\delta)$ can be thought of as the set of qubits placed on  $d$-simplices containing $\delta$. We say that qubits $\mathcal{Q}(\delta)$ are supported on $\delta$. By saying that an operator is supported on $\delta$ we mean that it is supported on the set $\mathcal{Q}(\delta)$, for example $X(\delta):=X(\mathcal{Q}(\delta))$.

A color code is a CSS subsystem code~\cite{Poulin2005,Bacon2006}. Recall that a CSS subsystem code is specified by its gauge group $\mathcal{G}$. Each $X$-type gauge group generator $X(G^x)$ consists of Pauli $X$ operators applied to qubits $G^x$; similarly for $Z$-type generators. The stabilizer group $\mathcal{S} \subset \mathcal{G}$ is the group generated by all Pauli operators $X(S^x)$ and $Z(S^z)$ contained in $\mathcal{G}$, which commute with every element of $\mathcal{G}$. Note that $-I\not\in\mathcal{S}$. The codewords are $+1$ eigenvectors of all elements of $\mathcal{S}$.

We define a $d$-dimensional color code~\cite{Bombin2013} on the lattice $\mathcal{L}$, where $d=\dim\mathcal{L}$, as the CSS subsystem code with of $X$- and $Z$-type gauge generators supported on $x$- and $z$-simplices in $\mathcal{L}$,
\begin{equation}
\label{gaugegeneratordefinition}
\mathcal{G} = \langle X(\delta), Z(\sigma)|\forall \delta\in\Delta'_x(\mathcal{L}),\sigma\in\Delta'_z(\mathcal{L}) \rangle,
\end{equation}
where $x+z\leq d-2$. The $X$- and $Z$-type generators of the stabilizer group $\mathcal{S}$ are supported on $(d-z-2)$- and $(d-x-2)$-simplices, namely
\begin{equation}
\mathcal{S}\! =\! \langle X(\delta), Z(\sigma)|\forall\delta\!\in\Delta'_{d - z - 2}(\mathcal{L}), \sigma\!\in\Delta'_{d - x - 2}(\mathcal{L}) \rangle.\label{stabilizergeneratordefinition}
\end{equation}

We refer to this code by $CC_{\mathcal{L}} (x,z)$. When context makes the lattice unambiguous, we sometimes use $CC_d (x,z)$ to emphasize the dimensionality of the lattice, $\dim\mathcal{L}=d$. Note that the generators of the gauge and stabilizer groups are supported on simplices which do not belong to $\partial\mathcal{L}$, the boundary of the lattice $\mathcal{L}$.

To illustrate the language introduced in this section, we revisit the two-dimensional color code described in Sections~\ref{2dOneQubit} and~\ref{DualLattice2d}. We begin with the lattice $\mathcal{L}$ shown in Fig.~\ref{fig:2DLatticeConstruction}d. Qubits are placed on $2$-simplices (triangular faces). Since $x+z\leq  \dim\mathcal{L} - 2 = 0$, there is only one color code on the two-dimensional lattice $\mathcal{L}$, namely $CC_{\mathcal{L}}(0,0)$, which is a stabilizer code. Stabilizer generators are associated with $0$-simplices (vertices). Note that no stabilizer generators are assigned to the three vertices belonging to the boundary of $\mathcal{L}$.

\subsection{Properties of the lattice}

Here we present some properties of any $(d+1)$-colorable homogeneous simplicial $d$-complex $\mathcal{L}$. We use these properties to verify that $CC_{\mathcal{L}}(x,z)$ is a valid code, and later that there is a transversal implementation of $\overline{R}_n$. We start with the following two lemmas

\begin{lemma}[Intersection]
Let $\delta$ and $\sigma$ be two simplices in $\mathcal{L}\setminus\partial\mathcal{L}$. If $\mathcal{Q}(\delta)\cap \mathcal{Q}(\sigma)\neq \emptyset$, then $\mathcal{Q}(\delta)\cap \mathcal{Q}(\sigma)=\mathcal{Q}(\tau)$, where $\tau$ is the smallest simplex containing both $\delta$ and $\sigma$.
\label{lemma:intersection}
\end{lemma}
\begin{proof}
If $\mathcal{Q}(\delta)\cap \mathcal{Q}(\sigma)\neq \emptyset$, then there exists $\epsilon\in\Delta_d(\mathcal{L})$ such that $\epsilon\supset \delta,\sigma$.   Let $ C=\col{\delta}\cup\col{\sigma}$ and set $\tau$ to be the unique $(| C|-1)$-simplex in $\epsilon$, colored with the set of colors $ C$. Clearly, $\tau$ is the smallest simplex containing $\delta$ and $\sigma$, and $\mathcal{Q}(\delta)\cap \mathcal{Q}(\sigma)=\mathcal{Q}(\tau)$.
\end{proof}

\begin{lemma}[Even Support]
Let $\delta$ be a $k$-simplex not belonging to the boundary of the lattice, $\delta\subset\Delta'_k(\mathcal{L})$, with $0 \leq k< d$. Then
\begin{equation}
|\mathcal{Q}(\delta)|\equiv 0 \mod 2.
\end{equation}
\label{lemma:evenfaces}
\end{lemma}

Before we prove the (Even Support) Lemma~\ref{lemma:evenfaces}, we explain its consequences. For $CC_{d}(x,z)$ to be a subsystem code, the stabilizer generators have to commute with each other, as well as with the gauge group generators. Notice that for two arbitrary $X$- and $Z$-type stabilizer generators to commute, the intersection of their supports has to contain even number of elements. Let $X$- and $Z$-type stabilizer generators be supported on $\delta\subset\Delta'_x(\mathcal{L})$ and $\sigma\subset\Delta'_z(\mathcal{L})$, respectively. If the intersection $\mathcal{Q}(\delta)\cap \mathcal{Q}(\sigma)$ is non-empty, then due to the (Intersection) Lemma~\ref{lemma:intersection} there exists a simplex $\tau$ such that $\mathcal{Q}(\delta)\cap \mathcal{Q}(\sigma)=\mathcal{Q}(\tau)$. Moreover, since $\delta$ is spanned by $x+1$ vertices and $\sigma$ by $z+1$ vertices, then $\tau$ is spanned by at most $x+z+2 \leq d$ vertices. Thus, $\tau$ is a $k$-simplex with $k<d$, and the (Even Support) Lemma ~\ref{lemma:evenfaces} applies, $|\mathcal{Q}(\delta)\cap \mathcal{Q}(\sigma)|=|\mathcal{Q}(\tau)|\equiv 0 \mod 2$, showing that $X(\delta)$ and $Z(\sigma)$ commute. The commutation of stabilizer generators with the gauge generators follows similarly.

We can obtain the (Even Support) Lemma~\ref{lemma:evenfaces} as a corollary of the following

\begin{lemma}[Disjoint Union]
Let $\mathcal{L}$ be a simplicial  $d$-complex which is $(d+1)$-colorable. Then, for a simplex $\delta\subset\mathcal{L}\setminus\partial\mathcal{L}$ and a chosen set of colors $C$, such that $\col{\delta}\subset C\subset\mathbb{Z}_{d+1}$, there exists a partition of the set of qubits supported on $\delta$ into a disjoint union of sets of qubits supported on $(|C|-1)$-simplices containing $\delta$, namely 
\begin{equation}
\mathcal{Q}(\delta)=\bigsqcup_{\substack{\sigma\supset\delta\\ \sigma\in\Delta'_{|C|-1}(\mathcal{L})\\  \col{\sigma}= C}} \mathcal{Q}(\sigma).
\label{DisjointUnionEq}
\end{equation}
\label{DisjointUnionLemma}
\end{lemma}
\begin{proof}
First note, that two different $k$-simplices $\delta_1$ and $\delta_2$ in $\mathcal{L}\setminus\partial\mathcal{L}$ colored with the same colors, $\col{\delta_1}=\col{\delta_2}$, cannot belong to the same $l$-simplex, $l\geq k$, thus do not share a qubit, $\mathcal{Q}(\delta_1)\cap \mathcal{Q}(\delta_2)=\emptyset$. Moreover, if $\mathcal{Q}(\epsilon)\subset \mathcal{Q}(\delta)$, where $\epsilon\in\Delta_d(\mathcal{L})$, then $\epsilon\supset\delta$ and there exists a unique simplex $\sigma\subset\epsilon$ colored with colors $B$. Since $\col{\sigma}= C\supset\col{\delta}$, then $\sigma\supset\delta$, which finishes the proof of the (Disjoint Union) Lemma~\ref{DisjointUnionLemma}.
\end{proof}

In particular, the set of qubits supported on any $k$-simplex $\delta$ in $\mathcal{L}\setminus\partial\mathcal{L}$ with $k<d$ can be decomposed as a disjoint union of qubits suppoted on $(d-1)$-simplices $\sigma$ containing $\delta$ and colored with a chosen set of $d$ colors, $C\supset\col{\delta}$. Notice, that $|\mathcal{Q}(\sigma)|=2$ for any $\sigma\in\Delta'_{d-1}(\mathcal{L})$, which immediately yields 
\begin{equation}
|\mathcal{Q}(\delta)|=\sum_{\substack{\sigma\supset\delta\\ \sigma\in\Delta'_{d-1}(\mathcal{L})\\ \col{\sigma}= C}} |\mathcal{Q}(\sigma)|\equiv 0 \mod 2,
\end{equation}
showing the (Even Support) Lemma~\ref{lemma:evenfaces}.

The property needed for the transversal implementation of the gate $\overline{R}_n$, presented in Section~\ref{sec:gates}, can be encapsultated in the following lemma

\begin{lemma}[Bipartition of Qubits]
The set of $d$-simplices in $\mathcal{L}$, $\Delta_d(\mathcal{L})$, is bipartite.
\label{lemma:qubitsbipartition}
\end{lemma}

Let us first explain the (Bipartition of Qubits) Lemma~\ref{lemma:qubitsbipartition} --- the $d$-simplices in $\mathcal{L}$ can be split into two disjoint sets, where $d$-simplices in the first set share $(d-1)$-faces only with $d$-simplices from the second set, and vice versa.  

\begin{proof}
First, construct a graph $G=(V,E)$ with the set of vertices $V=\Delta_d(\mathcal{L})$ and the set of edges $E=\Delta'_{d-1}(\mathcal{L})$. Two vertices $v,w\in V$ are connected by an edge $e\in E$ iff  $d$-simplices corresponding to $v$ and $w$ share a $(d-1)$-face corresponding to $e$. Since for all $\delta\in\Delta'_{d-2}(\mathcal{L})$ the (Even Support) Lemma~\ref{lemma:evenfaces} gives $|\mathcal{Q}(\delta)|\equiv 0 \mod 2$, and every cycle in $\mathcal{L}$ is contractible, we obtain that every cycle in the graph $G$ is even. Using the (Graph Bipartition) Lemma~\ref{lemma:Bipartition} we immediately obtain that $G$ is bipartite. This shows that the set of $d$-simplices in $\mathcal{L}$, which is equal to the set of qubits, $\Delta_d(\mathcal{L})=Q$, is bipartite.
\end{proof}

\section{Transversal gates in color codes}
\label{sec:gates}

As mentioned in the introduction, transversal gates are fault-tolerant. In this section, we first review some relevant features of a class of CSS subsystem codes, which includes the color codes defined in Sectionref{sec:Ddim}. Then, we examine transversal gates of codes in this class. We show that $\overline{\cnot}$ is transversal in any such code and under certain additional conditions the Hadamard and $\overline{R}_n$ can be transversal, too. Finally, we show that the additional conditions are satisfied by certain color codes.

\subsection{Subsystem codes}

A CSS subsystem code \cite{Poulin2005,Bacon2006} is specified by its gauge group $\mathcal{G}$, which is a subgroup of the Pauli group on physical qubits $Q$. Each $X$-type gauge group generator $X(G^x)$ consists of Pauli $X$ operators applied to qubits $G^x$; similarly for $Z$-type generators. The stabilizer group $\mathcal{S} \subseteq \mathcal{G}$ is the group generated by all Pauli operators $X(S^x)$ and $Z(S^z)$ contained in $\mathcal{G}$, which commute with every element of $\mathcal{G}$. Note that a stabilizer code is a special case of a subsystem code, for which $\mathcal{G} = \mathcal{S}$. The codewords are the $+1$ eigenvectors of all elements of $\mathcal{S}$. We say that two codewords are equivalent if they differ by application of a linear combination of elements of $\mathcal{G} \setminus \mathcal{S}$. This allows one to decompose the subspace of codewords into a tensor product of two spaces: \emph{logical} qubits and \emph{gauge} qubits. Elements of $\mathcal{G} \setminus \mathcal{S}$ have no effect on the state of the logical qubits, but may change that of the gauge qubits.

For a subsystem code, we say a unitary implements a \emph{logical gate} if it preserves the space of all codewords, and has an action on the logical qubits which is independent of any action on the gauge qubits. A logical gate $\overline{U}$ can be implemented on the logical qubits $\ket{\psi}$ as a \emph{bare} gate $U_{\text{bare}}$ which leaves gauge qubits $\ket{g}$ unchanged, $U_{\text{bare}}: |\psi \rangle | g \rangle \mapsto (\overline{U} |\psi \rangle ) | g \rangle$, or more generally as a \emph{dressed} gate $U_{\text{dressed}}$, which can affect the gauge qubits too, $U_{\text{dressed}}: |\psi \rangle | g \rangle \mapsto (\overline{U} |\psi \rangle ) | g' \rangle$. 

Consider the class of CSS subsystem codes which
\begin{itemize}
\item encode one logical qubit,
\item have bare logical $\overline{X}$ and $\overline{Z}$ implemented by $X(Q)$ and $Z(Q)$.
\end{itemize}
Note that these codes are defined on an odd number of physical qubits, $|Q|\equiv 1\mod 2$, since $\overline{X}$ and $\overline{Z}$ anticommute.

We can define a pair of inequivalent (and not normalized) codewords, which are representatives of logical $\ket{\overline{0}}$ and $\ket{\overline{1}}$, namely
\begin{eqnarray}
\label{Xgaugebasis1}
|\overline{0} \rangle |g_X\rangle &=& \sum_{ X(G) \in \mathcal{G}} X(G) \ket{\mathbf{0}},\\
\label{Xgaugebasis2}
|\overline{1} \rangle  |g_X\rangle &=& \overline{X} |\overline{0} \rangle  |g_X\rangle,
\end{eqnarray}
where $\ket{\mathbf{0}}$ is a state with every physical qubit set to $\ket{0}$, and $|g_X\rangle$ is a fixed state of the gauge qubits. One can verify that the states $|\overline{0} \rangle |g_X\rangle$ and $|\overline{1} \rangle  |g_X\rangle$ are $+1$ eigenstates of $\mathcal{S}$, and satisfy $\overline{Z}|\overline{0} \rangle |g_X\rangle = |\overline{0} \rangle |g_X\rangle$, $\overline{Z}|\overline{1} \rangle |g_X\rangle = -|\overline{1} \rangle |g_X\rangle$. They are also $+1$ eigenstates of every $X$-type generator of $\mathcal{G}$. All equivalent codewords can be generated from $|\overline{0} \rangle |g_X\rangle, |\overline{1} \rangle |g_X\rangle$ by application of a linear combination of elements from $\mathcal{G} \setminus \mathcal{S}$. An alternative pair of representatives of logical $\ket{\overline{0}}$ and $\ket{\overline{1}}$ is 
\begin{eqnarray}
\label{Zgaugebasis1}
|\overline{0} \rangle |g_Z\rangle &=& \sum_{ X(S) \in \mathcal{S}} X(S) \ket{\mathbf{0}},\\
\label{Zgaugebasis2}
|\overline{1} \rangle  |g_Z\rangle &=& \overline{X} |\overline{0} \rangle  |g_Z\rangle,
\end{eqnarray}
which are $+1$ eigenstates of all $Z$-type generators of $\mathcal{G}$.

\subsection{Transversal gates in subsystem codes}

Consider a CSS subsystem code with one logical qubit, and $\overline{X}$ and $\overline{Z}$ implemented by $X(Q)$ and $Z(Q)$. To check that a physical unitary $U$ implements a dressed logical gate $\overline{U}$ in such a code, one can verify its action on $|\overline{0} \rangle |g\rangle$, and $|\overline{1} \rangle  |g\rangle$ for every state $|g \rangle$ of the gauge qubits. Alternatively, it is sufficient to verify that $U$ has the correct action by conjugation on $\overline{X}$ and $\overline{Z}$, and that it preserves\footnote{Note that preservation of the gauge group under the action of a physical unitary $U$ is a sufficient, but not a necessary condition for $U$ to implement a dressed logical gate.} the gauge group $\mathcal{G}$. 

The logical gate $\overline{\cnot}$ can be implemented transversally between two identical copies of a CSS subsystem code by applying a physical gate $\cnot$ to every pair of corresponding qubits in the first and the second copy. This can be verified by checking that under conjugation by $\overline{\cnot}$, $\overline{X}\overline{I} \mapsto \overline{X}\overline{X}$, $\overline{I}\overline{X} \mapsto \overline{I}\overline{X}$, $\overline{Z}\overline{I} \mapsto \overline{Z}\overline{I}$, $\overline{I}\overline{Z} \mapsto \overline{Z}\overline{Z}$ and $\mathcal{G}\otimes\mathcal{G}$ is preserved\footnote{Notice, that generators of $\mathcal{G}\otimes\mathcal{G}$ are mapped under conjugation to another set of generators, namely $X(G)\otimes I(G) \mapsto X(G)\otimes X(G)$, $Z(G)\otimes I(G) \mapsto I(G)\otimes Z(G)$, $I(G)\otimes X(G) \mapsto I(G)\otimes X(G)$ and $I(G)\otimes Z(G) \mapsto Z(G)\otimes Z(G)$.}.

For a \emph{self-dual} CSS subsystem code, namely a code with $X$- and $Z$-type gauge group generators supported on the same sets of qubits, $\mathcal{G}=\langle X(G_i),Z(G_i)\rangle$, a dressed logical Hadamard gate can be implemented transversally as $\overline{H}=H(Q)$. To see this, observe that under conjugation by $H(Q)$, $\overline{X} \mapsto \overline{Z}$, $\overline{Z} \mapsto \overline{X}$, $X(G)\mapsto Z(G)$ and $Z(G)\mapsto X(G)$, and thus $\mathcal{G}$ is preserved.

The last logical gate we analyze is $\overline{R}_n=\textrm{diag}\left(1,e^{\frac{2\pi i}{2^n}}\right)$, for an integer $n>0$. We aim to implement $\overline{R}_n$ transversally as a bare logical gate by applying the same single-qubit unitary to some subset $T\subset Q$ of the physical qubits, and applying that unitary's inverse to the rest of the qubits $T^c := Q \setminus T$. Specifically, we now prove that $\overline{R}_n$ is implemented by $R= R^k_n (T) R^{-k}_n (T^c)$, for some suitably chosen $k\in\{1,2,\ldots,2^n-1\}$, provided that $T$ and $\mathcal{G}$ satisfy
\begin{equation}
\forall X(G)\in \mathcal{G}: |T\cap G|\equiv |T^c\cap G| \mod 2^n .
\label{eq:conditiononT}
\end{equation}
First, pick $k$ such that 
\begin{equation}
k(|T| - |T^c|)\equiv 1\mod 2^n.
\label{eq:definek}
\end{equation}
The existence of $k$ is guaranteed by Bezout's lemma, since $|Q|$ is odd, $|T| - |T^c| = 2|T| - |Q|  \equiv 1 \mod 2$, and thus \linebreak[4] $\text{gcd}(2|T| - |Q|,2^n)=1$. 
Noting that $R_n^{\pm k}|0\rangle = |0 \rangle$ and $R_n^{\pm k} X  = e^{\pm \frac{2\pi i k}{2^n}} X R_n^{\mp k}$, we obtain
\begin{eqnarray}
R |\overline{0}\rangle |g_X\rangle  &=& \sum_{X(G)\in \mathcal{G}} R^k_n(T) R^{-k}_n(T^c ) X(G) |\mathbf{0}\rangle\\
&=& \sum_{X(G)\in \mathcal{G}} e^{\frac{2\pi i k}{2^n}|T\cap G|} e^{-\frac{2\pi i k}{2^n}|T^c \cap G|} X(G)  \ket{\mathbf{0}}\ \ \ \ \ \\
 &=& \sum_{X(G)\in \mathcal{G}} X(G)  |\mathbf{0}\rangle=|\overline{0}\rangle  |g_X\rangle  ,\\
  R |\overline{1}\rangle |g_X\rangle 
 & =& R^k_n(T) R^{-k}_n(T^c )  X(Q) |\overline{0} \rangle |g_X\rangle \\
  &=& e^{\frac{2\pi i k}{2^n}|T|} e^{-\frac{2\pi i k}{2^n}|T^c|} X(Q) R |\overline{0} \rangle |g_X\rangle \\
  &=& e^{\frac{2\pi i }{2^n} } X(Q) | \overline{0} \rangle |g_X\rangle = e^{\frac{2\pi i }{2^n} } | \overline{1} \rangle |g_X\rangle,
\end{eqnarray}
which shows that $R$ correctly implements logical $\overline{R}_n$ when the gauge qubits are in the state $|g_X\rangle$. However, all other states of the gauge qubits can be reached by application of $Z$-type operators from $\mathcal{G} \setminus \mathcal{S}$, which all commute with $R$ (since it is diagonal in the $Z$-basis). Therefore for any state $|g \rangle$ of the gauge qubits, it must be that $R: |\overline{0}\rangle |g \rangle \mapsto  |\overline{0}\rangle |g \rangle, |\overline{1}\rangle |g \rangle \mapsto  e^{\frac{2\pi i }{2^n} } |\overline{1}\rangle |g \rangle$, verifying that $R$ implements the bare logical gate $\overline{R}_n$. 

It may not be obvious that there exists a set $T\subset Q$ satisfying (\ref{eq:conditiononT}) for a given code. In the later parts of this paper we show an explicit construction of $T$ for color codes in  $d$ dimensions, with $n \leq d$. Notice that condition~(\ref{eq:conditiononT}) can be inferred from the following condition
\begin{equation}
\left| T\cap \bigcap_{i=1}^m G_i\right| \equiv \left| T^c\cap \bigcap_{i=1}^m G_i\right| \mod 2^{n-m+1},
\label{eq:conditiononT2}
\end{equation}
where $m=1,\ldots,n$ and $\{ X(G_1),\ldots, X(G_m)\}$ is any subset of the $X$-type generators of the gauge group $\mathcal{G}$. To see the implication (\ref{eq:conditiononT2})$\implies$(\ref{eq:conditiononT}) notice, that for any $X(G)\in\mathcal{G}$, we can write it as product of generators, namely $X(G)=\prod ^{m}_{i=1}X(G_i)$. Then
\begin{equation}
G=G_1\veebar G_2 \veebar\ldots \veebar G_m,
\end{equation}
where we used the symmetric difference of sets, $A\veebar B:=(A\setminus B)\cup(B\setminus A)$. Using the Inclusion-Exclusion Principle for symmetric difference\footnote{For sets $A_1,A_2,\ldots,A_m$, we have $|A_1\veebar A_2\veebar \ldots\veebar A_m|=\sum_i |A_i|-2\sum_{i\neq j} |A_i\cap A_j|+\ldots +(-2)^{m-1}|A_1\cap A_2\cap\ldots\cap A_m|.$}
we obtain 
\begin{eqnarray}
|T\cap G|&=& |T\cap (G_1\veebar G_2\veebar\ldots\veebar G_m )|\\
&=& \sum_i |T\cap G_i |-
2\sum_{i\neq j}|T\cap (G_i\cap G_j )|+ \\
&& 4\sum_{i\neq j\neq k}|T\cap (G_i\cap G_j\cap G_k )|-\ldots  \\
&&  +(-2)^{m-1}|T\cap (G_1 \cap G_2\cap\ldots\cap G_m )|,
\end{eqnarray}
and a similar expression for $|T^c \cap G|$.
Clearly, if condition (\ref{eq:conditiononT2}) holds, then $|T\cap G|- |T^c \cap G|\equiv 0 \mod 2^n$, showing (\ref{eq:conditiononT}). Moreover, condition (\ref{eq:conditiononT2}) is easier to verify than condition~(\ref{eq:conditiononT}), since we only need to check it for the $X$-type generators of $\mathcal{G}$, rather than for every $X$-type element of $\mathcal{G}$.

We can summarize the discussion of how to implement transversal $\overline{R}_n$ in the following lemma
\begin{lemma}[Sufficient Condition]
\label{lem:sufficientcondition}
Consider a CSS subsystem code encoding one logical qubit. Let the code be defined on a set of physical qubits $Q$, where $|Q|$ is odd and with bare logical operators $\overline{X}=X(Q)$ and $\overline{Z}=Z(Q)$. If there exists $T\subset Q$, such that for any $m=1,\ldots,n$:
\begin{equation}
\left| T\cap\bigcap_{i=1}^m G_i\right| \equiv \left| T^c\cap\bigcap_{i=1}^m G_i\right| \mod 2^{n-m+1},
\label{eq:sufficientcondition}
\end{equation}
for every subset $\{ X(G_1),\ldots, X(G_m)\}$ of the $X$-type gauge generators of the code, then 
\begin{equation}
R= R^k_n (T) R^{-k}_n (T^c)
\end{equation}
implements logical $\overline{R}_n$, where $k$ is a solution to $k(|T| - |T^c|)\equiv 1\mod 2^n$ and $T^c = Q \setminus T$.
\end{lemma}

\subsection{Transversal implementation of $\overline{R}_n$ in color code}

Here we show how to implement the logical gate $\overline{R}_n$ transversely in the color code $CC_{\mathcal{L}} (x,z)$, for any integer $n \leq \dim(\mathcal{L})/(x+1)$. One applies $R=R^k_n(T) R^{-k}_n(T^c)$ for some integer $k$, where $T$ and its compliment $T^c = Q\setminus T$ correspond to the bipartite decomposition of qubits $Q$ specified in the (Bipartition of Qubits) Lemma~\ref{lemma:qubitsbipartition}. We make use of the following property
\begin{lemma}[Property of $T$]
For any $m$-simplex $\sigma$ in $\mathcal{L}\setminus\partial\mathcal{L}$ with $m<d$
\begin{equation}
|T \cap \mathcal{Q}(\sigma)| = |T^c \cap \mathcal{Q}(\sigma)|.
\end{equation}
\label{lemmadual}
\end{lemma}

\begin{proof}
By the choice of the set $T$, every $(d-1)$-simplex $\delta$ has one qubit in $T$, and one qubit in $T^c = Q\setminus T$, which is equivalent to $|T \cap \mathcal{Q}(\delta)|=|T^c \cap \mathcal{Q}(\delta)|$. Using the (Disjoint Union) Lemma~\ref{DisjointUnionLemma}, we can decompose the set of qubits $\mathcal{Q}(\sigma)$ supported on an $m$-simplex $\sigma$, where $m<d$, as a disjoint union of qubits supported on $(d-1)$-simplices colored with a chosen set of  $d$ colors $ C\supset\col{\sigma}$ , and then we immediately obtain
\begin{eqnarray}
|T\cap \mathcal{Q}(\sigma)| &&-|T^c\cap \mathcal{Q}(\sigma)| =\\
\hfill &&\sum_{\substack{\delta\supset\sigma\\ \delta\in\Delta'_{d-1}(\mathcal{L})\\ \col{\delta}= C}} |T \cap \mathcal{Q}(\delta)| - |T^c \cap \mathcal{Q}(\delta)|=0,\ \ \ \ 
\end{eqnarray}
which shows the (Property of $T$) Lemma~\ref{lemmadual}.

\end{proof}
Note that (\ref{eq:sufficientcondition}) in the (Sufficient Condition) Lemma~\ref{lem:sufficientcondition} follows form the (Property of $T$) Lemma~\ref{lemmadual}. To see this, observe first that every stabilizer generator $X(\delta_i)$ is supported on a $x$-simplex $\delta_i$, thus $G_i = \mathcal{Q}(\delta_i)$ and we obtain
\begin{equation}
\bigcap_{i=1}^m \mathcal{Q}(\delta_i)=\emptyset\text{\ \ \ or\ \ \ }\bigcap_{i=1}^m \mathcal{Q}(\delta_i)=\mathcal{Q}(\tau),
\end{equation}
where $\tau$ is a simplex colored with colors $C=\bigcup_{i=1}^m \col{\delta_i}$, such that $\tau\supset \delta_1,\ldots,\delta_m$. The case of an empty intersection is trivial. Since $|\col{\delta_i}|=x+1$, then obviously $| C| \leq  m(x+1) \leq d$, and thus $\tau$ is at most $(d-1)$-simplex. Using the (Property of $T$) Lemma~\ref{lemmadual} we obtain that for any $m =1,...,n$:
\begin{eqnarray}
\left|T\cap \bigcap_{i=1}^m \mathcal{Q}(\delta_i)\right|&&- 
\left|T^c \cap \bigcap_{i=1}^m \mathcal{Q}(\delta_i)\right|=\\
&&|T\cap \mathcal{Q}(\tau)| - |T^c \cap \mathcal{Q}(\tau)|=0,
\end{eqnarray}
which implies (\ref{eq:sufficientcondition}). The (Sufficient Condition) Lemma~\ref{lem:sufficientcondition} implies that $R$ implements the logical $\overline{R}_n$. In particular, one can implement $\overline{R}_d$ using the code $CC_d(0,d-2)$, since $\dim\mathcal{L}=d$, $x=0$ and thus $\left\lfloor\frac{\dim\mathcal{L}}{x+1}\right\rfloor = d$. 

\section{Universal transversal gates with color codes}
\label{sec:universal}

A finite set of gates which is universal can be used to implement any logical unitary, with arbitrary precision. In particular, due to the Solovay-Kitaev~\cite{Kitaev1997, Nielsen2010} theorem, the number of applied gates scales poly-logarithmically with the precision of approximation. Note that the set $\{\overline{H}, \overline{\cnot}, \overline{R}_n \}$ is universal for any integer $n>2$.

In this section, we show how to achieve a universal transversal gate set with color codes by using the technique of \emph{gauge fixing} to switch between different codes. This technique allows one to take advantage of the transversally implementable gates for different color codes. We first illustrate the method with a simple example of two 15-qubit codes~\cite{Paetznick2013,Anderson2014}. Then, we define a partial order between color codes. One can switch between color codes which are comparable to implement a universal gate set in three or higher dimensions. 

\subsection{Switching between codes using gauge fixing}
\label{sec:15qubit}

First, let us define matrices $H_1$ and $H_2$ given by
\begin{eqnarray}
H_1=\left( \begin{array}{ccccccccccccccc}
1 & 1 & 1 & 1 & 1 & 1 & 1 & 1 & 0 & 0 & 0 & 0 & 0 & 0 & 0 \\
1 & 1 & 1 & 1 & 0 & 0 & 0 & 0 & 1 & 1 & 1 & 1 & 0 & 0 & 0 \\
1 & 1 & 0 & 0 & 1 & 1 & 0 & 0 & 1 & 1 & 0 & 0 & 1 & 1 & 0 \\
1 & 0 & 1 & 0 & 1 & 0 & 1 & 0 & 1 & 0 & 1 & 0 & 1 & 0 & 1 \\
\end{array} \right)\!\!,\ \ \\
H_2=\left( \begin{array}{ccccccccccccccc}
1 & 1 & 1 & 1 & 0 & 0 & 0 & 0 & 0 & 0 & 0 & 0 & 0 & 0 & 0 \\
1 & 1 & 0 & 0 & 1 & 1 & 0 & 0 & 0 & 0 & 0 & 0 & 0 & 0 & 0 \\
1 & 0 & 1 & 0 & 1 & 0 & 1 & 0 & 0 & 0 & 0 & 0 & 0 & 0 & 0 \\
1 & 1 & 0 & 0 & 0 & 0 & 0 & 0 & 1 & 1 & 0 & 0 & 0 & 0 & 0 \\
1 & 0 & 1 & 0 & 0 & 0 & 0 & 0 & 1 & 0 & 1 & 0 & 0 & 0 & 0 \\
1 & 0 & 0 & 0 & 1 & 0 & 0 & 0 & 1 & 0 & 0 & 0 & 1 & 0 & 0 \\
\end{array} \right)\!\!.\ \ 
\end{eqnarray}
Moreover, for a binary matrix $M$, we define $M^X$ to be a matrix obtained from $M$ by the following substitutions, $0\mapsto I$ and $1\mapsto X$. Similarly for $M^Z$, we substitute $0\mapsto I$ and $1\mapsto Z$.  Let $\mathscr{C}_A$ be the stabilizer code with the stabilizer group $\mathcal{S}_A$ generated by rows of $H^X_1$, $H^Z_1$ and $H^Z_2$, which we denote by
\begin{equation}
\mathcal{S}_A=\langle H^X_1, H^Z_1 ,H^Z_2 \rangle.
\end{equation}
Let $\mathscr{C}_B$ be the subsystem code with the stabilizer group $\mathcal{S}_B$ and the gauge group $\mathcal{G}_B$ chosen as follows
\begin{equation}
\mathcal{S}_B=\langle H^X_1, H^Z_1 \rangle,\ \ \mathcal{G}_B=\langle H^X_1, H^X_2, H^Z_1 ,H^Z_2 \rangle.
\end{equation}
We can consider both codes $\mathscr{C}_A$ and $\mathscr{C}_B$ to be defined on the same $15$ physical qubits. One can check that $\mathscr{C}_A$ represents the $[[15,1,3]]$ quantum Reed-Muller (stabilizer) code~\cite{MacWilliams1977,SteaneQRM,Anderson2014} and $\mathscr{C}_B$ is a $[[15,1,3]]$ (subsystem) code, which can be thought of as the $[[15,7,3]]$ Hamming code, with six of the seven logical qubits treated as gauge qubits. Note also that $\mathcal{S}_B \subset \mathcal{G}_A =\mathcal{S}_A$ and $\mathcal{G}_B$ has $X$- and $Z$-type generators supported on the same qubits (i.e. $\mathscr{C}_B$ is a self-dual subsystem code). 

Since the $X$-type generators of $\mathcal{G}_B$ coincide with the $X$-type generators of $\mathcal{S}_A$, the codewords of $\mathscr{C}_A$ and $\mathscr{C}_B$ are the same when the latter has a gauge state $|g_Z\rangle$. In other words, codewords $|\bar{0}\rangle$, $|\bar{1}\rangle$ for $\mathscr{C}_A$ are the same as codewords $|\bar{0}\rangle |g_Z\rangle$, $|\bar{1}\rangle |g_Z\rangle$ for $\mathscr{C}_B$, as defined in Eqs.~(\ref{Zgaugebasis1})~and~(\ref{Zgaugebasis2}). On the other hand the codewords $|\bar{0}\rangle |g_X\rangle$, $|\bar{1}\rangle |g_X\rangle$ for $\mathscr{C}_B$ (as defined in Eqs.~(\ref{Xgaugebasis1})~and~(\ref{Xgaugebasis2})), are not valid codewords for $\mathscr{C}_A$.

Now we show that $R_3^{\otimes 15}$ implements $\overline{R}_3$ transversally in $\mathscr{C}_A$. Consider any three of the four $X$-type generators for $\mathcal{G}_A$, and specify their support on subsets of qubits $G_1$, $G_2$, $G_3$, which correspond to rows of $H_1$. One can verify that $|G_a| =8 \equiv 0 \mod 2^3$, $|G_a \cap G_b| =4 \equiv 0 \mod 2^2$, and $|G_a \cap G_b \cap G_c| =2 \equiv 0 \mod 2$, where $\{ a,b,c \}=\{ 1,2,3 \}$. Therefore by the (Sufficient Condition) Lemma~\ref{lem:sufficientcondition}, and by setting $T$ to be an empty set, $T = \emptyset$, we see that $R_3^{\otimes 15}$ implements $\overline{R}_3$ transversally in the code $\mathscr{C}_A$. In contrast for the code $\mathscr{C}_B$, the extra $X$-type generators in $\mathcal{G}_B \setminus \mathcal{G}_A$ do not satisfy these conditions, and thus one cannot show that $\overline{R}_3$ is implemented transversally in $\mathscr{C}_B$.

It is straightforward to verify that $\overline{H}$ is implemented transversally by $H^{\otimes 15}$ in $\mathscr{C}_B$. It swaps $X$ and $Z$ on any physical qubit, and therefore acts on the representative states as $H^{\otimes 15}: |\psi \rangle |g_Z\rangle \mapsto (\overline{H}|\psi \rangle )|g_X\rangle$. Since the state of the gauge qubits has changed, $H^{\otimes 15}$ is a dressed implementation of $\overline{H}$ in $\mathscr{C}_A$. Clearly, $H^{\otimes 15}$ does not implement $\overline{H} $ in $\mathscr{C}_A$, since it takes the state $|\psi \rangle |g_Z\rangle\in\mathscr{C}_A$ to $(\overline{H}|\psi \rangle )|g_X\rangle\not\in\mathscr{C}_A$.

To implement $\overline{H}$ fault-tolerantly in $\mathscr{C}_A$, we use the technique of \emph{gauge fixing}. First, one should apply $H^{\otimes 15}$, resulting in mapping $ |\psi \rangle |g_Z\rangle$ to $(\overline{H}|\psi \rangle )|g_X\rangle$, which is a codeword of $\mathscr{C}_B$, but not of $\mathscr{C}_A$. Then, to switch from code $\mathscr{C}_B$ to $\mathscr{C}_A$, one should sequentially measure each of the six $Z$-type stabilizer generators generated by rows of $H^Z_2$, i.e. those in $\mathcal{S}_A \setminus \mathcal{S}_B$. Note that it is possible to fault-tolerantly measure the stabilizer generators in any stabilizer code \cite{Nielsen2010}. If the measurement reveals that a particular $Z$-type generator is not satisfied, then one should apply an $X$-type Pauli operator which commutes with all generators in $H^Z_2$ and $H^Z_1$, except for the violated stabilizer generator (with which it must anticommute). Such an $X$-type Pauli operator always exists. Following this application, the $Z$-type generator will no longer be violated. Therefore, after this is carried out for all six generators in $H^Z_2$, the state will have changed from $(\overline{H}|\psi \rangle )|g_X\rangle$ to $(\overline{H}|\psi \rangle )|g_Z\rangle$, as required. Specifically, we use the term \emph{gauge fixing} to refer to the process of measuring and setting the gauge qubits to a desired state. 

To recap, in the $[[15,1,3]]$ Reed-Muller code $\mathscr{C}_A$, one can implement $\overline{H}$ fault-tolerantly with the following procedure
\begin{equation}
\ket{\psi}\ket{g_Z} \xmapsto{H^{\otimes 15}} (\overline{H}\ket{\psi})\ket{g_X} \xmapsto{\text{gauge fixing}} (\overline{H}\ket{\psi})\ket{g_Z}.
\end{equation}
In combination with the transversal gates of $\mathscr{C}_A$, this allows one to implement a fault-tolerant universal gate set $\{\overline{H},\overline{\textrm{\cnot}},\overline{R}_3 \}$. We will repeat essentially the same procedure for color codes later.

\subsection{Partial order of color codes}

Given a $d$-dimensional lattice $\mathcal{L}$, $\dim\mathcal{L}=d$, satisfying Conditions~\ref{cond1}~and~\ref{cond2}, we can catalog all color codes defined on $\mathcal{L}$. Namely, a pair of integers $x,z\geq 0$, such that $x+z\leq d-2$, corresponds to a color code, denoted as $CC_{\mathcal{L}} (x,z)$, with $X$- and $Z$-type gauge generators supported on $x$- and $z$-simplices. Note that the $X$- and $Z$-type stabilizer generators of $CC_{\mathcal{L}} (x,z)$ are supported on $(d-2-z)$-simplices and $(d-2-x)$-simplices, respectively. In two dimensions, $d=2$, there is only one color code, $CC_2 (0,0)$ --- a stabilizer code, with both $X$- and $Z$-type stabilizer generators supported on $0$-simplices, whereas in three dimensions, $d=3$, there are three color codes, $CC_3 (1,0)$, $CC_3 (0,1)$ --- stabilizer codes, and $CC_3 (0,0)$ --- a subsystem code. 

One can define a partial order for subsystem color codes defined on the same lattice $\mathcal{L}$ if each codeword of code $\mathscr{C}$ is also a codeword of the other code $\mathscr{C}'$. In particular, we say that $\mathscr{C} \succ \mathscr{C}'$ holds if
\begin{itemize}
\item $\mathscr{C}$ and $\mathscr{C}'$ encode the same number of logical qubits, with identical bare logical Pauli operators,
\item the gauge group $\mathcal{G}$ of $\mathscr{C}$ is contained in the gauge group $\mathcal{G}'$ of $\mathscr{C}'$, $\mathcal{G} \subset \mathcal{G}'$.
\end{itemize}
Note that $\mathcal{G} \subset \mathcal{G}'$ implies $\mathcal{S}' \subset \mathcal{S}$, thus any codeword of $\mathscr{C}$ is also a codeword of $\mathscr{C}'$, and since the bare Pauli operators for the logical qubit are the same in both codes, it actually represents the same logical codeword in both codes. Observe, that the partial order we have just defined can be succinctly expressed as
\begin{equation}
CC_{\mathcal{L}} (x,z)\prec CC_{\mathcal{L}} (x',z') \iff x\leq x' \wedge z\leq z',
\end{equation}
as illustrated in Fig.~\ref{fig:colorcodesfamily}. This follows from the observation that due to the (Disjoint Union) Lemma~\ref{DisjointUnionLemma} the $X$-type gauge generators of $CC_{\mathcal{L}} (x,z)$, which are supported on $x$-simplices, can be expressed as the product of the $X$-type gauge generators of $CC_{\mathcal{L}} (x',z')$ supported on $x'$-simplices, since $x \leq x'$. Similarly for $Z$-type gauge generators. We represent the family of color codes in Fig.~\ref{fig:colorcodesfamily}, and show their partial order using arrows.

\begin{figure}[h!]
\includegraphics[width=0.35\textwidth]{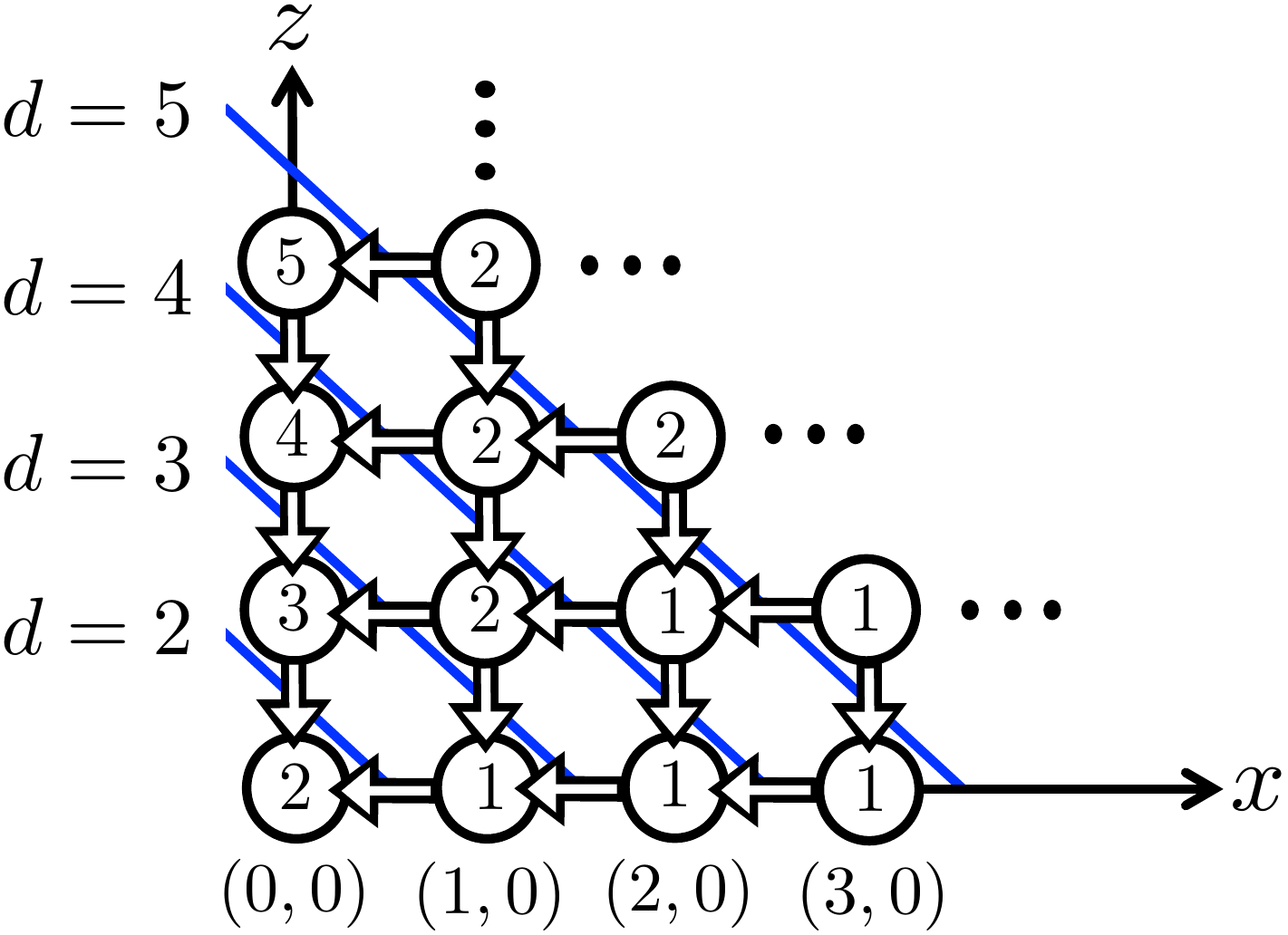}
\caption{Family of color codes. For a given lattice $\mathcal{L}$, only color codes below the $d^{\text{\,th}}$ diagonal line can be realized, where $d=\dim\mathcal{L}$ and the point $(x,z)$ corresponds to the color code $CC_{\mathcal{L}}(x,z)$. This constraint holds, since $x$ and $z$ have to satisfy $x+z\leq d-2$. An arrow from code $\mathscr{C}$ to $\mathscr{C}'$ indicates partial order between them, $\mathscr{C} \succ \mathscr{C}'$. The number placed at $(x,z)$ indicates the maximum gate $\overline{R}_n$ which can be implemented transversally with the stabilizer color code $CC_{d}(x,z)$, with $d=x+z+2$, resulting in $n=\left\lfloor \frac{d}{x+1} \right\rfloor$.}
\label{fig:colorcodesfamily} 
\end{figure}

\subsection{Universal fault-tolerant gate set in color codes}

Here we apply the techniques just discussed to color codes defined on the same lattice $\mathcal{L}$. One can switch back and forth between two codes which are comparable, $CC_{\mathcal{L}} (x,z)\prec CC_{\mathcal{L}} (x',z')$, as follows
\begin{itemize}
\item $CC_{\mathcal{L}} (x,z) \mapsto CC_{\mathcal{L}} (x',z')$: one does nothing, since codewords of $CC_{\mathcal{L}} (x,z)$ are codewords of $CC_{\mathcal{L}} (x',z')$, 
\item $CC_{\mathcal{L}} (x',z') \mapsto CC_{\mathcal{L}} (x,z)$: one can view the codewords of $CC_{\mathcal{L}} (x,z)$ as those for $CC_{\mathcal{L}} (x',z')$ with the additional gauge qubits present in $CC_{\mathcal{L}} (x,z)$ set to a particular state. To switch, one fixes the state of the additional gauge qubits to the appropriate state.
\end{itemize}

Given a three-dimensional lattice $\mathcal{L}$, $\dim\mathcal{L}=3$, one can implement a universal gate set starting with a code $CC_{\mathcal{L}}(0,1)$. As explained earlier, one can transversally perform the logical $\overline{\cnot}$ and $\overline{R_3}$ on that code. To form a universal gate set, it suffices to also implement logical $\overline{H}$. This gate cannot be implemented transversally in $CC_{\mathcal{L}}(0,1)$, but can be achieved in $CC_{\mathcal{L}}(0,0)$. Note that $CC_3(0,0) \prec CC_3(0,1)$, therefore any codeword in $CC_3(0,1)$ is a valid codeword in $CC_3(0,0)$. In particular, we can think of $\ket{\psi}\in CC_3(0,1)$ as $\ket{\psi}\ket{g}\in CC_3(0,0)$, where $\ket{g}$ is a state of the gauge qubits of $CC_3(0,0)$. By applying $H(Q)$ we perform the logical $\overline{H}$ on the logical qubits of $CC_3(0,0)$, which also changes the state of the gauge qubits, namely
\begin{equation}
H(Q)\left(\ket{\psi}\ket{g}\right)=\left(\overline{H}\ket{\psi}\right)\ket{g'}.
\end{equation}
Note that the resulting codeword $\left(\overline{H}\ket{\psi}\right)\ket{g'}\in CC_3(0,0)$ is not a valid codeword of $CC_3(0,1)$, since the gauge qubits are in the state $\ket{g'}\neq \ket{g}$. To return to $CC_3 (0,1)$, one needs to \emph{fix the gauge qubits} to the correct state, namely $\ket{g'}\mapsto \ket{g}$, and we obtain a codeword $\overline{H}\ket{\psi}\ket{g}\in CC_3(0,1)$. Since $CC_3(0,1)$ is a stabilizer code, it is possible to measure and correct the violated stabilizers in a fault-tolerant way, just as in~Section~\ref{sec:15qubit}. Therefore, to fix the gauge, one should first measure all $Z$-type stabilizer generators supported on 1-simplices, and then apply the appropriate $X$-type Pauli operators in order to correct any violated stabilizer generators. After this, assuming no errors have occurred, all the stabilizer generators for $CC_3(0,1)$ are satisfied.

To summarize, we can perform the logical $\overline{H}$ on $CC_3(0,1)$ by first applying $H(Q)$ and subsequently fixing the gauge to return to the codespace of $CC_3(0,1)$,
\begin{equation}
\ket{\psi}\ket{g'}  \xmapsto{H(Q)}  \left(\overline{H}\ket{\psi}\right)\ket{g'} \xmapsto{\textrm{  gauge fixing  }} \left(\overline{H}\ket{\psi}\right)\ket{g}.
\end{equation}

Since $\overline{\textrm{\cnot}}$ and $\overline{R}_3$ can be performed transversally in $CC_3(0,1)$, one can fault-tolerantly implement a universal gate-set  $\{\overline{H},\overline{\textrm{\cnot}},\overline{R}_3 \}$ in $CC_3(0,1)$. This procedure can be directly generalized to fault-tolerantly implement the universal gate set $\{\overline{H},\overline{\textrm{\cnot}},\overline{R}_d \}$ with the code $CC_d(0,d-2)$ in  $d$ dimensions.

\section{Acknowledgements}

We would like to thank H\'{e}ctor Bomb\'{i}n for introducing us to color codes and taking the time to explain his results. We would like to thank Jeongwan Haah, Beni Yoshida, Olivier Landon-Cardinal, Gorjan Alagic and John Preskill for helpful comments on the manuscript. We thank Fernando Pastawski for pointing out bipartition as a possible construction of the set $T$. We acknowledge funding provided by the Institute for Quantum Information and Matter, an NSF Physics Frontiers Center with support of the Gordon and Betty Moore Foundation (Grants No. PHY-0803371 and PHY-1125565).

%

\appendix*
\section{Examples of color codes}

\subsection{Construction of a lattice in  $d$ dimensions}
\label{LatticeConstruction}

A recipe to obtain a lattice $\mathcal{L}$ satisfying the Conditions~\ref{cond1}~and~\ref{cond2} required to define color codes in  $d$ dimensions is as follows (see Fig.~\ref{fig:construction2d} for an example in $d=2$).

\begin{figure}[h!]
\includegraphics[width=0.41\textwidth]{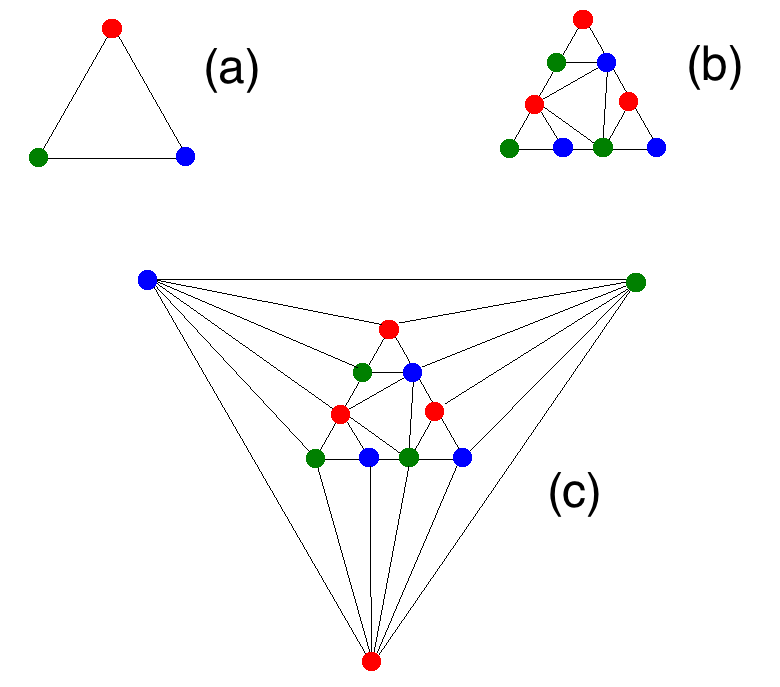}
\caption{Construction of a color code in 2D with spatially local (stabilizer) generators. (a) Take a $2$-simplex $\delta$, with vertices colored in red, green and blue. (b) Divide $\delta$ into ``smaller'' simplices with matching colors. This is a $3$-colorable homogeneous simplicial $2$-complex $\mathcal{K}$. (c) Place $\mathcal{K}$ inside a $2$-simplex $\tau$ and attach $2$-simplices between $\tau$ and $\mathcal{K}$. The resulting homogeneous simplicial $2$-complex $\mathcal{L}$ is $3$-colorable, and thus we can define a color code on the lattice $\mathcal{L}$.}
\label{fig:construction2d} 
\end{figure} 

\begin{enumerate}
\item Start with a  $d$-simplex $\delta$, with vertices which are colored with $d+1$ colors $\mathbb{Z}_{d+1}$.
\item Construct a homogeneous simplicial  $d$-complex $\mathcal{K}$ from $\delta$ by dividing $k$-faces of $\delta$ into $k$-simplices. We also require that the coloring is preserved, i.e. every $k$-face $\sigma\subset\delta$ colored with $ C=\col{\sigma}$ is divided into $k$-simplices colored with $ C$ and the whole complex $\mathcal{K}$ is $(d+1)$-colorable.
\item Place the  $d$-complex $\mathcal{K}$ inside a  $d$-simplex $\tau$ colored with $\mathbb{Z}_{d+1}$.
\item For every $k$-face $\rho\subsetneq\tau$ and for every $(d-k-1)$-simplex $\omega\subset \mathcal{K}$ obtained from a $(d-k-1)$-face $\sigma\subset\delta$ with complementary colors, $\col{\omega}=\mathbb{Z}_{d+1}\setminus \col{\rho}$, attach a  $d$-simplex spanned by $\rho$ and $\omega$.
\item Choose $\mathcal{L}$ to be the collection of all  $d$-simplices added in Step 4, together with simplices belonging to $\mathcal{K}$ and $\tau$. This can be used to define a color code on the lattice $\mathcal{L}$ as specified in Section~\ref{sec:Ddim}.   
\end{enumerate}

Note that in the above recipe, step 2 is not fully specified. Any homogeneous simplicial $d$-complex $\mathcal{K}$ obtained from a $d$-simplex $\delta$ will work, as long as $\mathcal{K}$ is $(d+1)$-colorable. Such lattices always exist --- below we give an explicit example of a family of lattices in any dimension $d\geq 2$. Following steps 3-5, we always obtain a lattice on which we can define a color code in $d$ dimensions. 

There is a systematic construction of a family of (fractal) color codes in $d$ dimensions, for which there is an explicit recipe for $\mathcal{K}$. The resulting codes neither have spatially local generators nor have macroscopic distance, and do not result in color codes, which are topological stabilizer codes. The prescription is as follows.
\begin{enumerate}
\item The first member is defined on the lattice $\mathcal{L}_1$, obtained from the recipe by setting $\mathcal{K}$ to be a $d$-simplex.
\item The $i+1$ member of the family is defined on the lattice $\mathcal{L}_{i+1}$, obtained from the recipe by setting $\mathcal{K}=\mathcal{L}_{i}$.
\end{enumerate}
The first three members of the family of the two-dimensional (fractal) color codes are illustrated in Fig.~\ref{fig:fractalcode}.

In Ref.~\cite{Bombin2013}, a systematic construction in two and three dimensions for families of color codes with spatially local generators is presented. In two dimensions, $\mathcal{K}$ is chosen to be a part of triangular lattice (as in Fig.~\ref{fig:construction2d}), whereas in three dimensions $\mathcal{K}$ is a part of the BCC lattice. Bomb\'{i}n's constructions result in topological color codes.

\begin{figure}[h!]
\includegraphics[width=0.45\textwidth]{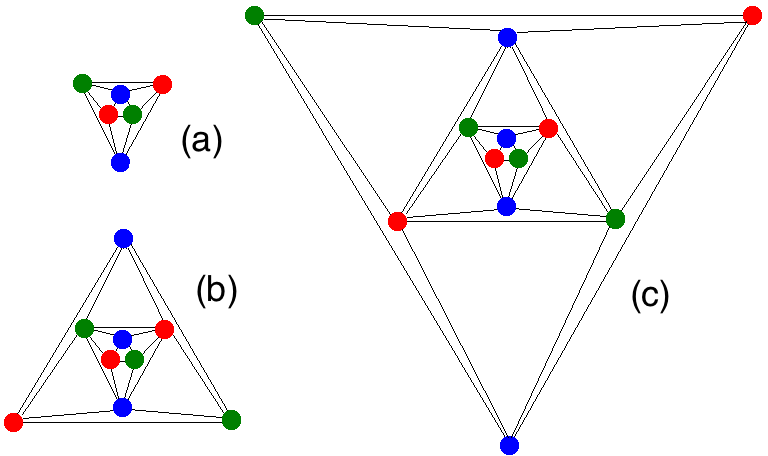}
\caption{The family of (fractal) color codes in two dimensions. The first three members of the family --- two-dimensional color codes encoding one logical qubits using (a) 7, (b) 13 and (c) 19 physical qubits.}
\label{fig:fractalcode} 
\end{figure}

\subsection{Quantum Reed-Muller codes as color codes}

There exists a family of codes known as the quantum Reed-Muller codes~\cite{MacWilliams1977,SteaneQRM,Anderson2014}. Here we are concerned with the subfamily of quantum Reed-Muller codes with members labeled uniquely by an integer $m\geq 3$ with parameters $[[2^m-1,1,3]]$, i.e. encoding one logical qubit into $2^m-1$ physical qubits, with a distance of three. We denote by $QRM(m)$ the $m^{\text{th}}$ member of this subfamily. These codes are defined in terms of matrices $M_i$ satisfying the recursion relations
\begin{equation}
M_1 = (1),\ M_{i+1}=
\left(\begin{array}{ccc}
M_i & 0& M_i \\
0\ldots 0  & 1& 1\ldots 1 \\
\end{array}\right).
\end{equation}
Note that the set of columns of $M_m$ is the set of all non-zero binary vectors of length $m$. By $M^{\perp}_m$ we denote a matrix dual to $M_m$, namely a matrix with rows being a basis of the kernel of $M_m$. Clearly, $M_m (M^{\perp}_m)^T=0$. We can define $QRM(m)$ as the stabilizer code with the stabilizer group $\mathcal{S}_m$ generated by rows of $M_m$ and $M^{\perp}_m$ with $0$'s and $1$'s replaced by $I$'s and $X$'s or $Z$'s , namely
\begin{equation}
\mathcal{S}_m=\langle M^X_m, (M^{\perp}_m)^Z \rangle.
\end{equation}
\vspace*{1pt}
We now show that $QRM(m)$ is the same as the stabilizer color code $CC_{m-1}(0,m-3)$ obtained from the construction described in Appendix~\ref{LatticeConstruction} by taking the simplicial complex $\mathcal{K}$ to be a $(m-1)$-simplex $\delta$, $\mathcal{K}=\delta$. In other words, $QRM(m)$ is equal to the first member of the (fractal) color code family in $m-1$ dimensions (see Fig.~\ref{fig:fractalcode} (a) for $m=3$ case). In particular, $QRM(3)$ is Steane's $7$-qubit code and $QRM(4)$ is the $15$-qubit Reed-Muller code. 

\begin{figure}[h!]
\includegraphics[width=0.45\textwidth]{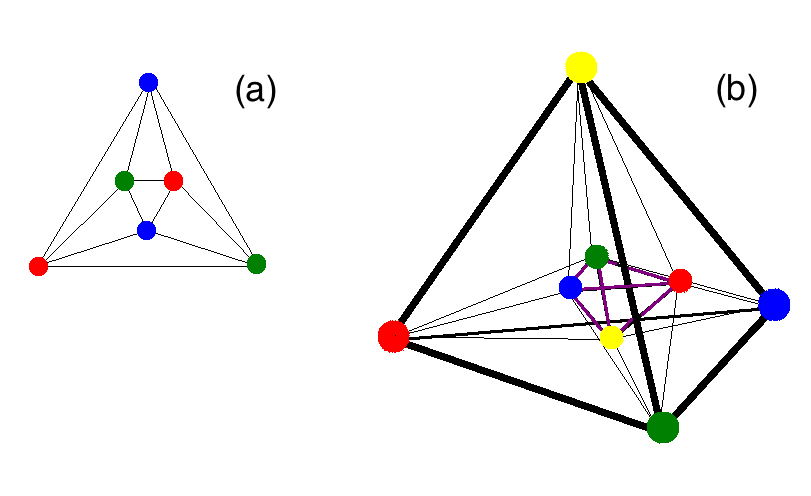}
\caption{Quantum Reed-Muller code $QRM(m)$ as a special case of a (stabilizer) color code $CC_{m-1}(0,m-3)$ for (a) $m=3$ --- Steane's $7$-qubit code, and (b) $m=4$ --- the $15$-qubit Reed-Muller code. Steane's code with all the possible transversal gates have recently been implemented experimentally~\cite{Nigg2014}.
}
\label{fig:reedmuller} 
\end{figure}

To prove this equivalence, it is sufficient to show that there is a one-to-one identification of physical qubits of $QRM(m)$ with those of $CC_{m-1}(0,m-3)$ such that the logical Pauli operators $\overline{X}$ and $\overline{Z}$ are identical, and that the $X$-type stabilizer generators are identical. Note that this completely specifies the stabilizer group $\mathcal{S}$, since the $Z$-type generator matrix is a dual to the $X$-type generator matrix. In particular, we show that the $X$-type generator matrix $M'_m$ for $CC_{m-1}(0,m-3)$ is the same as $M_m$ for $QRM(m)$ up to a permutation of columns.

Using the construction described in Appendix~\ref{LatticeConstruction}, and taking the simplicial complex $\mathcal{K}=\delta$, where $\delta$ is a $(m-1)$-simplex, results in a lattice $\mathcal{L}$, with $\dim\mathcal{L} = m-1$. The total number of  $(m-1)$-simplices in $\mathcal{L}$ is $2^{m}-1$. This is because we attach $(m-1)$-simplices between every $(k-1)$-face $\rho\subset\tau$, for every $1\leq k\leq m-1$, and the $(m-k-1)$-face $\sigma\subset\delta$ colored with the complementary colors, $\col{\sigma}=\mathbb{Z}_{d+1}\setminus\col{\rho}$. We can pick a subset of $k$ vertices of $\tau$ in $m \choose k$ different ways and thus the number of newly attached $(m-1)$-simplices is ${m\choose 1}+{m\choose 2}+\ldots +{m\choose m-1}=2^{m}-2$. Therefore, including a qubit placed at $\delta$, there are exactly $2^m-1$ physical qubits in $CC_{\mathcal{L}}(0,m-3)$. On the other hand, there are exactly $m$ vertices in $\mathcal{L}\setminus\partial\mathcal{L}$, and thus there are $m$ $X$-type stabilizer generators in $CC_{\mathcal{L}}(0,m-3)$. The weight of a column in $M'_m$, corresponding to a qubit supported on a $(m-1)$-simplex $\pi$, is given by the number of $X$-type stabilizer generators supported on that qubit, i.e. the number of vertices belonging to $\pi$ but not to $\partial\mathcal{L}$. There are exactly $m \choose m-k$ $(m-1)$-simplices containing $k$ vertices not belonging to $\partial\mathcal{L}$ and each of them contains different  set of $k$ vertices. Thus, there are ${m\choose k}$ different columns of weight $k$ in $M'_m$ and the only way this can occur is if the columns of $M'_m$ are the set of all non-zero binary vectors of length $m$. Thus, up to a relabeling of physical qubits, $M'_m$ and $M_m$ are identical. Also note that the logical operators of both codes are $\overline{X} = X(Q)$ and $\overline{Z} = Z(Q)$. Therefore the codes are the same.
\vfill
\eject

\end{document}